\documentclass[10pt,conference,romanappendices]{IEEEtran}
\usepackage[cmex10]{amsmath}
\usepackage{amssymb}
\usepackage{amsfonts}

\newtheorem{definition}{Definition} 
\newtheorem{proposition}{Proposition}
\newtheorem{lemma}{Lemma}

\newtheorem{remark}{Remark}

\newenvironment{proof}{\hspace{8pt}\ti{Proof:}}{~~~~QED}

\newcommand{\bm}[1]{{\mbox{\boldmath $#1$}}}
\newcommand{\pnt}[1]{{\mbox{\boldmath $#1$}}}
\newcommand{\cof}[2]{\mbox{$#1_{\boldsymbol{#2}}$}}

\newcommand{\nGz}[2]{$G_{non-\{z\}}$}

\newcommand{\mi}[1]{\mathit{#1}}
\newcommand{\ti}[1]{\textit{#1}}
\newcommand{\tb}[1]{\textbf{#1}}

\newcommand{\ttt}{\>\>\>}

\newcommand{\Tt}{\>\>}

\newcommand{\Sup}[2]{\mbox{$#1^\mi{#2}$}}

\newcommand{\prob}[2]{\mbox{$\exists{#1} [#2]$}}
\newcommand{\pprob}[2]{\exists{#1} [#2]}

\newcommand{\Comment}[1]{}

\newcommand{\tr}[3]{\mbox{(\pnt{#1_{#2}},\dots,\pnt{#1_{#3}})}}
\newcommand{\Tr}[3]{(\pnt{#1_{#2}},\!\dots,\!\pnt{#1_{#3}})}

\newcommand{\impl}[2]{\mbox{$#1 \rightarrow #2$}}
\newcommand{\ks}{\mbox{$\xi$}}
\newcommand{\KS}{\mbox{$\xi$}~}

\newcommand{\pqe}[4]{\mbox{$\prob{#1}{#2 \wedge #4} \equiv #3 \wedge \prob{#1}{#4}$}}

\newcommand{\Abs}[2]{\mbox{$\mathbb{#1}_{#2}$}}

\mathchardef\mhyphen="2D

\newcommand{\di}[1]{\mbox{$\mi{Diam}(#1)$}}
\newcommand{\pp}{\mbox{$\mi{ProveProp}$}}
\newcommand{\PP}{\mbox{$\mi{ProveProp}$}~}

\newcommand{\fc}{\mbox{$\mi{ProveProp*}$}}
\newcommand{\FC}{\mbox{$\mi{ProveProp*}$}~}
\newcommand{\s}[1]{\mbox{$\{#1\}$}}

\usepackage{wrapfig}
\usepackage{graphicx}
\usepackage{enumerate}
\usepackage[hidelinks]{hyperref}
\begin{document}

\title{Property Checking Without Inductive Invariants}

\author{\IEEEauthorblockN{Eugene Goldberg} 
\IEEEauthorblockA{
eu.goldberg@gmail.com}}

\maketitle

\begin{abstract}
We introduce a procedure for proving safety properties.  This
procedure is based on a technique called Partial Quantifier
Elimination (PQE).  In contrast to complete quantifier elimination, in
PQE, only a part of the formula is taken out of the scope of
quantifiers. So, PQE can be dramatically more efficient than complete
quantifier elimination.  The appeal of our procedure is
twofold. First, it can prove a property without generating an
inductive invariant. Second, it employs depth-first search and so can
be used to find deep bugs.
\end{abstract}

\section{Introduction}
%
%
\subsection{Motivation}
\label{ssec:motiv}
Property checking is an important part of hardware verification. (In
this paper, by property checking we mean verification of \ti{safety}
properties.) Arguably, the most popular way to prove a property is to
build an inductive invariant. While the recent SAT-based methods for
generating inductive invariants have enjoyed great
success~\cite{ken03,ic3}, this approach has some problems.  First, for
some properties, \ti{compact} inductive invariants may not
exist. Second, even if such inductive invariants exist, they may be
very \ti{hard to find}.  Third, if a property does not hold due to a
\ti{deep} bug, running an algorithm for building an inductive
invariant may not be the best way to find this bug.  In this paper, we
describe a method for proving properties without generation of
inductive invariants. This method employs depth-first search and so
can be useful for finding deep bugs.

%
%
\subsection{Partial quantifier elimination}
\label{ssec:pc_by_pqe}
This paper is a part of our effort to develop a technique called
partial quantifier elimination (PQE)\cite{hvc-14}. In contrast to
regular (i.e. ``complete'') quantifier elimination, only a part of the
formula is taken out of the scope of quantifiers in PQE. The appeal of
PQE is twofold.  First, it provides a language for \ti{incremental}
computing. Second, PQE can be dramatically \ti{more efficient} than
complete quantifier elimination.

In our work on PQE we combine bottom-up and top-down approaches.  The
bottom-up part is to develop algorithms for efficiently solving
PQE~\cite{hvc-14,fmcad12,fmcad13,cert_tech_rep}. The top-down part is
to create PQE based methods for solving verification problems.  For
instance, we have described such methods for SAT
solving~\cite{hvc-14,south_korea,cert_tech_rep}, equivalence
checking~\cite{fmcad16}, model checking~\cite{tech_rep_pc_lor},
testing~\cite{cert_tech_rep,tech_rep_crr}, checking the completeness
of specification~\cite{con_props}. This paper is an addition to the
top-down part of our research. Namely, it describes a PQE based
algorithm for property checking that does not generate an inductive
invariant.

%
%
\subsection{Problem we consider}
Let \KS be a transition system specified by transition relation
$T(S,S')$ and formula $I(S)$ describing initial states. Here $S$ and
$S'$ are sets of variables specifying the present and next states
respectively. Let \pnt{s} be a state i.e. an assignment to
$S$. Henceforth, by an assignment \pnt{q} to a set of variables $Q$ we
mean a \ti{complete} assignment unless otherwise stated i.e. all
variables of $Q$ are assigned in \pnt{q}.

Let $P(S)$ be a property of \ks. We will call a state \pnt{s} a
\smallskip
\bm{P}\tb{-state} if $P(\pnt{s})=1$.  We will refer to a
$\overline{P}$-state (i.e. a state where $P$ fails) as a \tb{bad
  state}.  The problem we consider is to check if a bad state is
reachable in \ks.  We will refer to this problem as \ti{the safety
  problem}.  Usually, the safety problem is solved by finding an
\tb{inductive invariant} i.e. a formula $K(S)$ such that \impl{K}{P}
and \impl{K(S) \wedge T(S,S')}{K(S')}.

%
%
\subsection{Property checking without inductive invariants}
\label{ssec:pc_no_inv}
In this paper, we consider an approach where the safety problem is
solved without generation of an inductive invariant. We will refer to
a system with initial states $I$ and transition relation $T$ as an
\bm{(I,T)}\tb{-system}. Let \di{I,T} denote the \tb{reachability
  diameter} of an $(I,T)$-system. That is $n=\di{I,T}$ means that
every state of this system is reachable in at most $n$ transitions.
One can partition the problem of checking if property $P$ holds into
two subproblems.
\begin{enumerate}
\item Find the value of \di{I,T}.
\vspace{3pt}
\item Check if a $\overline{P}$-state is reachable in $n$ transitions
  where $n \leq \di{I,T}$.
\end{enumerate}
We describe a procedure called \PP that solves the two subproblems
above.  We will refer to the first subproblem as the \tb{RD problem}
where RD stands for Reachability Diameter.

\PP is formulated in terms of PQE, which has the following advantages.
First, the RD problem above is solved without generating the set of
all reachable states. Second, the property $P$ is proved without
generation of an inductive invariant. Third, due to using PQE, \PP
performs depth-first search, which is beneficial for finding deep
bugs. Importantly, as we mentioned above, PQE can be much more
efficient than complete quantifier elimination.

To prove a property $P$ true, \PP has to consider traces of length up
to \di{I,T}. This may slow down property checking for systems with a
large diameter. We describe a variation of \PP called \FC that can
prove a property without generation of long traces. \FC achieves
faster convergence by expanding the set of initial states with
\mbox{$P$-states} i.e. by replacing $I$ with \Sup{I}{exp} such that
\impl{I}{\Sup{I}{exp}} and \impl{\Sup{I}{exp}}{P}.

%
%
\subsection{Contributions and structure of the paper}
The contribution of this paper is twofold. First, we present a
procedure for finding the reachability diameter without computing the
set of all reachable states. Second, we describe a procedure for
proving a property without generating an inductive invariant.

This paper is structured as follows.  We recall PQE in
Section~\ref{sec:pqe_def}. Basic definitions and notation conventions
are given in Section~\ref{sec:defs}.  Section~\ref{sec:gist} describes
how one can look for a counterexample and solve the RD problem by PQE.
In Section~\ref{sec:dfs}, we show that PQE enables depth-first search
in property checking.  The \PP and \FC procedures are presented in
Sections~\ref{sec:prove_prop} and~\ref{sec:prove_prop*} respectively.
Section~\ref{sec:background} provides some background. We make
conclusions in Section~\ref{sec:conclusions}.

\section{partial quantifier elimination}
\label{sec:pqe_def}
In this paper, by a quantified formula we mean one with
\ti{existential} quantifiers. We assume that all formulas are
propositional formulas in Conjunctive Normal Form (\tb{CNF}). The
latter is a conjunction of \tb{clauses}, a clause being a disjunction
of literals.

Given a quantified formula \prob{W}{A(V,W)}, the problem of
\ti{quantifier elimination} is to find a quantifier-free formula
$A^*(V)$ such that $A^* \equiv \prob{W}{A}$.  Given a quantified
formula \prob{W}{A(V,W) \wedge B(V,W)}, the problem of \tb{Partial
  Quantifier Elimination} (\tb{PQE}) is to find a quantifier-free
formula $A^*(V)$ such that \pqe{W}{A}{A^*}{B}.  Note that formula $B$
remains quantified (hence the name \ti{partial} quantifier
elimination). We will say that formula $A^*$ is obtained by
\tb{taking} \pnt{A} \tb{out of the scope of quantifiers} in \prob{W}{A
  \wedge B}. We will call $A^*$ a \tb{solution} to the PQE problem
above.

Let $G(V)$ be a formula implied by $B$. Then \pqe{W}{A}{A^* \wedge
  G}{B} implies that $\prob{W}{A \wedge B} \equiv A^* \wedge
\prob{W}{B}$ . In other words, clauses implied by the formula that
remains quantified are \ti{noise} and can be removed from a solution
to the PQE problem. So, when building $A^*$ by resolution it is
sufficient to use only the resolvents that are descendants of clauses
of $A$. For that reason, in the case formula $A$ is much smaller than
$B$, PQE can be dramatically faster than complete quantifier
elimination.  In this paper, we do not discuss PQE solving. This
information can be found in \cite{hvc-14,cert_tech_rep}.

\section{Definitions and notation}
\label{sec:defs}
\subsection{Basic definitions}
\label{subsec:defs}
%
%
\begin{definition}
Let \KS be an $(I,T)$-system. An assignment \pnt{s} to state variables
$S$ is called \tb{a state}.  A sequence of states \tr{s}{0}{n} is
called a \tb{trace}. This trace is called \tb{valid} if
\begin{itemize}
\item $I(\pnt{s_0})=1$,
\item $T(\pnt{s_i},\pnt{s_{i+1}}) = 1$ where $i=0,\dots,n-1$.
\end{itemize}
\end{definition}

Henceforth, we will drop the word \ti{valid} if it is obvious from the
context whether a trace is valid.
%
%
\begin{definition}
Let \tr{s}{0}{n} be a valid trace of an $(I,T)$-system.  State
\pnt{s_n} is said to be \tb{reachable} in this system in $n$
transitions.
\end{definition}
%
%
\begin{definition}
Let \KS be an $(I,T)$-system and $P$ be a property to be checked.  Let
\tr{s}{0}{n} be a valid trace such that
\begin{itemize}
\item every state \pnt{s_i}, $i=0,\dots,n-1$ is a $P$-state i.e. $P(\pnt{s_i})=1$.
\item state \pnt{s_n} is a $\overline{P}$-state i.e. $P(\pnt{s_n})=0$. 
\end{itemize}
Then this trace is called a \tb{counterexample} for property $P$.
\end{definition}
\begin{remark}
We will use the notions of a CNF formula $C_1 \wedge .. \wedge C_p$
and the set of clauses \s{C_1,\dots,C_p} \ti{interchangeably}. In
particular, the fact that $F$ has no clauses (i.e. $F = \emptyset $)
also means $F \equiv 1$ and vice versa.
\end{remark}

%
%
\subsection{Some notation conventions}
\label{subsec:notation}
\begin{itemize}
\item $S_j$ denotes the state variables of $j$-th time frame.
\item \Abs{S}{j} denotes $S_0 \cup \dots \cup S_j$. 
\item $T_{j,j+1}$ denotes $T(S_j,S_{j+1})$.  
\item \Abs{T}{j} denotes $T_{0,1} \wedge \dots
T_{j-1,j}$.  
\item  $I_0$ and $I_1$ denote $I(S_0)$ and $I(S_1)$ respectively.
\end{itemize}

\section{Property Checking By PQE}
\label{sec:gist}
In this section, we explain how the \PP procedure described in this
paper proves a property without generating an inductive invariant.
This is achieved by reducing the RD problem and the problem of finding
a bad state to PQE.  To simplify exposition, we consider systems with
stuttering. This topic is discussed in
Subsection~\ref{ssec:stuttering}. There we also explain how one can
introduce stuttering by a minor modification of the system. The main
idea of \PP and two propositions on which it is based are given in
Subsection~\ref{ssec:main_idea}.

%
%
\subsection{Stuttering}
\label{ssec:stuttering}
Let \KS denote an $(I,T)$-system. The \PP procedure we describe in
this paper is based on the assumption that \KS has the \tb{stuttering}
feature. This means that $T(\pnt{s},\pnt{s})$=1 for every state
\pnt{s} and so \KS can stay in any given state arbitrarily long.  If
\KS does not have this feature, one can introduce stuttering by adding
a combinational input variable $v$.  The modified system \KS works as
before if $v=1$ and remains in its current state if $v=0$. (For the
sake of simplicity, we assume that \KS has only sequential
variables. However, one can easily extend explanation to the case
where \KS has combinational variables.)

On the one hand, introduction of stuttering does not affect the
reachability of a bad state and does not affect the value of
\di{I,T}. On the other hand, stuttering guarantees that \KS has two
nice properties.  First, $\prob{S}{T(S,S')} \equiv 1$ holds since for
every next state \pnt{s'}, there is a ``stuttering transition'' from
\pnt{s} to \pnt{s'} where \pnt{s} = \pnt{s'}.  Second, if a state is
unreachable in \KS in $n$ transitions it is also unreachable in $m$
transitions if $m < n$. Conversely, if a state is reachable in \KS in
$n$ transitions, it is also reachable in $m$ transitions where $m >
n$.

%
%
\subsection{Solving the RD problem and finding a bad state by PQE}
\label{ssec:main_idea}
As we mentioned in the introduction, one can reduce property checking
to solving the RD problem and checking whether a bad state is
reachable in $n$ transitions where $n \leq \di{I,T}$. In this
subsection, we show that one can solve these two problems by PQE.

The RD problem is to compute $\di{I,T}$. It reduces to finding the
smallest $n$ such that the sets of states reachable in $n$ and $(n+1)$
transitions are identical. The latter, as
Proposition~\ref{prop:diam_pqe} below shows, comes down to checking if
formula $I_1$ is redundant in \prob{\Abs{S}{n}}{I_0 \wedge I_1 \wedge
  \Abs{T}{n+1}} i.e. whether $\prob{\Abs{S}{n}}{I_0 \wedge I_1 \wedge
  \Abs{T}{n+1}} \equiv \prob{\Abs{S}{n}}{I_0 \wedge \Abs{T}{n+1}}$. If
$I_1$ is redundant, then $\di{I,T} \leq n$.  This is a special case of
the PQE problem. Instead, of finding a formula $H(S_{n+1})$ such that
$\prob{\Abs{S}{n}}{I_0 \wedge I_1 \wedge \Abs{T}{n+1}} \equiv H \wedge
\prob{\Abs{S}{n}}{I_0 \wedge \Abs{T}{n+1}}$ one just needs to decide
if a tautological formula $H$ (i.e. $H \equiv 1$) is a solution to the
PQE problem above.

Here is an informal explanation of why redundancy of $I_1$ means
$\di{I,T} \leq n$.  The set of states reachable in $n+1$ transitions
is specified by \prob{\Abs{S}{n}}{I_0 \wedge \Abs{T}{n+1}}. Adding
$I_1$ to $I_0 \wedge \Abs{T}{n+1}$ shortcuts the initial time frame
and so \prob{\Abs{S}{n}}{I_0 \wedge I_1 \wedge \Abs{T}{n+1}} specifies
the set of states reachable in $n$ transitions. Redundancy of $I_1$
means that the sets of states reachable in $n+1$ and $n$ transitions
are the same. Proposition~\ref{prop:diam_pqe} states that the
intuition above is correct.
%
%
\begin{proposition}
\label{prop:diam_pqe}
Let \KS be an $(I,T)$-system. Then $\di{I,T} \leq n$ iff formula $I_1$
is redundant in $\pprob{\Abs{S}{n}}{I_0 \wedge I_1 \wedge
  \Abs{T}{n+1}}$ (i.e.  iff $\pprob{\Abs{S}{n}}{I_0 \wedge
  \Abs{T}{n+1}} \equiv \pprob{\Abs{S}{n}}{I_0 \wedge I_1 \wedge
  \Abs{T}{n+1}}$).
\end{proposition}

Proposition~\ref{prop:bad_state} below shows that one can look for
bugs by checking if $I_1$ is redundant in $\prob{\Abs{S}{n}}{I_0
  \wedge I_1 \wedge \Abs{T}{n+1} \wedge \overline{P}}$ i.e. similarly
to solving the RD problem.

%
%

\begin{proposition}
\label{prop:bad_state}
Let \KS be an $(I,T)$-system and $P$ be a property of \ks. No
$\overline{P}$-state is reachable in $(n+1)$-th time frame for the
first time iff $I_1$ is redundant in $\prob{\Abs{S}{n}}{I_0 \wedge I_1
  \wedge \Abs{T}{n+1} \wedge \overline{P}}$.
\end{proposition}

\section{Depth-first Search By PQE}
\label{sec:dfs}
In this section, we demonstrate that Proposition~\ref{prop:diam_pqe}
enables depth-first search when solving the RD-problem. Namely, we
show that one can prove that $\di{I,T} > n$ without generation of all
states reachable in $n$ transitions. In a similar manner, one can show
that Proposition~\ref{prop:bad_state} enables depth-first search when
looking for a bad state.  Hence it facilitates finding deep bugs.

Proposition~\ref{prop:diam_pqe} entails that proving $\di{I,T} > n$
comes down to showing that $I_1$ is not redundant in
\prob{\Abs{S}{n}}{I_0 \wedge I_1 \wedge \Abs{T}{n+1}}. This can be done
by 
\begin{enumerate}[a)]
\item generating a clause $C(S_{n+1})$ implied by $I_0 \wedge I_1
  \wedge \Abs{T}{n+1}$ 
\item finding a trace that satisfies $I_0 \wedge \Abs{T}{n+1} \wedge
  \overline{C}$
\end{enumerate}
Let \tr{s}{0}{n+1} be a trace satisfying $I_0 \wedge \Abs{T}{n+1}
\wedge \overline{C}$. On the one hand, conditions a) and b) guarantee
that $\prob{\Abs{S}{n}}{I_0 \wedge \Abs{T}{n+1}} \neq
\prob{\Abs{S}{n}}{I_0 \wedge I_1 \wedge \Abs{T}{n+1}}$ under
\pnt{s_{n+1}}. So $I_1$ is not redundant in \prob{\Abs{S}{n}}{I_0
  \wedge I_1 \wedge \Abs{T}{n+1}}. On the other hand, condition a)
shows that \pnt{s_{n+1}} is not reachable in $n$ transitions (see
Proposition~\ref{prop:dfs} of the appendix) and condition b) proves
\pnt{s_{n+1}} reachable in $(n+1)$-transitions.

Note that satisfying conditions a) and b) above \ti{does not require
  breadth-first search} i.e. computing the set of all states reachable
in $n$ transitions. In particular, clause $C$ of condition a) can be
found by taking $I_1$ out of the scope of quantifiers in
\prob{\Abs{S}{n}}{I_0 \wedge I_1 \wedge \Abs{T}{n+1}} i.e. by solving
the PQE problem. In the context of PQE-solving, condition b) requires
$C$ not be a ``noise'' clause implied by $I_0 \wedge \Abs{T}{n+1}$
i.e. the part of the formula that remains quantified (see
Section~\ref{sec:pqe_def}). In this paper, we assume that one employs
PQE algorithms that may generate noise clauses.  So, to make sure that
condition b) holds, one must prove $I_0 \wedge \Abs{T}{n+1} \wedge
\overline{C}$ satisfiable.

\section{\PP procedure}
\label{sec:prove_prop}
In this section, we describe procedure \pp.  As we mentioned in
Subsection~\ref{ssec:pc_no_inv}, when proving that a safety property
$P$ holds, \PP solves the two problems below.
\begin{enumerate}
\item Find the value of \di{I,T} (i.e. solve
  the RD problem).
\vspace{3pt}
\item Check that no $\overline{P}$-state is reachable in $n$
  transitions where $n \leq \di{I,T}$.
\end{enumerate}

\PP returns a counterexample if $P$ does not hold, or the value of
\di{I,T} if $P$ holds. To simply find the value of \di{I,T}, one can
call \PP with the trivial property $P$ (i.e. $P \equiv 1$).

A description of how \PP solves the RD problem is given in
Subsection~\ref{ssec:find_diam}. Solving the RD problem is accompanied
in \PP by checking if a bad state is reached. That is the problems
above are solved by \PP \ti{together}. A description of how \PP checks
if a counterexample exists is given in Subsection~\ref{ssec:find_cex}.
The pseudo-code of \PP is described in
Subsections~\ref{ssec:find_bug_pcode} and ~\ref{ssec:two_routines}.

%
%
\subsection{Finding diameter}
\label{ssec:find_diam}
Proposition~\ref{prop:diam_pqe} entails that proving $\di{I,T} \leq n$
comes down to showing that formula $I_1$ is redundant in
\prob{\Abs{S}{n}}{I_0 \wedge I_1 \wedge \Abs{T}{n+1}}. To this end,
\PP builds formulas $H_1,\dots,H_n$. Here $H_1$ is a subset of clauses
of $I_1$ and formula $H_i(S_i)$, $i=2,\dots,n$ is obtained by
resolving clauses of $I_0 \wedge I_1 \wedge \Abs{T}{i}$.  One can view
formulas $H_1,\dots,H_n$ as a result of ``pushing'' clauses of $I_1$
and their descendants (obtained by resolution) to later time frames.
The main property satisfied by these formulas is that
$\prob{\Abs{S}{n-1}}{I_0 \wedge I_1 \wedge \Abs{T}{n}} \equiv H_n
\wedge \prob{\Abs{S}{n-1}}{I_0 \wedge \Abs{H}{n-1} \wedge \Abs{T}{n}}$
where $\Abs{H}{n-1} = H_1 \wedge \dots \wedge H_{n-1}$.

\PP starts with $n=1$ and $H_1 = I_1$. Then it picks a clause $C$ of
$H_1$ and finds formula $H_2(S_2)$ such that $\prob{\Abs{S}{1}}{I_0
  \wedge H^*_1 \wedge C \wedge \Abs{T}{2}} \equiv H_2 \wedge
\prob{\Abs{S}{1}}{I_0 \wedge H^*_1 \wedge \Abs{T}{2}}$ where $H^*_1 =
H_1 \setminus \s{C}$. Formula $H_1$ is replaced with
$H^*_1$. Computing $H_2$ can be viewed as pushing $C$ to the second
time frame.  If $H_2 \equiv 1$, \PP picks another clause of $H_1$ and
computes $H_2$ for this clause. If $H_2 \equiv 1$ for every clause of
$H_1$, then eventually $H_1$ becomes empty (and so $H_1 \equiv
1$). This means that $I_1$ is redundant in \prob{\Abs{S}{1}}{I_0
  \wedge H_1 \wedge \Abs{T}{2}} and $\di{I,T} \leq 2$. If $H_2$ is
not empty, then \PP picks a clause $C$ of $H_2$ and builds formula
$H_3(S_3)$ such that $\prob{\Abs{S}{2}}{I_0 \wedge H_1 \wedge H^*_2
  \wedge C \wedge \Abs{T}{3}} \equiv H_3 \wedge \prob{\Abs{S}{1}}{I_0
  \wedge H_1 \wedge H^*_2 \wedge \Abs{T}{3}}$ where $H^*_2 = H_2
\setminus \s{C}$. Formula $H_2$ is replaced with $H^*_2$. If $H_3
\equiv 1$ for every clause of $H_2$, then $H_2$ becomes empty and \PP
picks a new clause of $H_1$. This goes on until all descendants of
$I_1$ proved redundant.

The procedure above is based on the observation that if $n =
\di{I,T}$, any clause $C(S_{n+1})$ that is a descendant of $I_1$ is
implied by $I_0 \wedge \Abs{T}{n+1}$. So one does not need to add
clauses depending on $S_{n+1}$ to make formula $H_n$ redundant. In
other words, pushing the descendants of $I_1$ to later time frames
inevitably results in making them redundant. The value of \di{I,T} is
given by the largest index $n$ among the time frames where a
descendant of $I_1$ was not redundant yet.

%
%
\subsection{Finding a counterexample}
\label{ssec:find_cex}
Every time \PP replaces a clause of $H_n(S_n)$ with formula
$H_{n+1}(S_{n+1})$, it checks if a $\overline{P}$-state is
reached. This is done as follows. Recall that $H_{n+1}$ satisfies
$\prob{\Abs{S}{n}}{I_0 \wedge I_1 \wedge \Abs{T}{n+1}} \equiv H_{n+1}
\wedge \prob{\Abs{S}{n}}{I_0 \wedge \Abs{H}{n} \wedge \Abs{T}{n+1}}$.
Let $C$ be a clause of $H_{n+1}$. Assume for the sake of simplicity
that \PP employs a noise-free PQE-solver. Then clause $C$ is derived
\ti{only} if it is implied by $I_0 \wedge I_1 \wedge \Abs{T}{n+1}$ but
not $I_0 \wedge \Abs{T}{n+1}$. This means that the very fact that $C$
is derived guarantees that there is \ti{some} state that is reached in
$(n+1)$-th iteration for the first time.  So if a $\overline{P}$-state
\pnt{s} falsifies $C$, there is a chance that \pnt{s} is reachable.
Now, suppose that $H_{n+1}$ is generated by a PQE-solver that may
generate noise clauses but the amount of noise is small. In this case,
generation of clause $C$ above still implies that there is a
significant probability of \pnt{s} being reachable.

A bad state falsifying $C$ is generated by \PP as an assignment
satisfying $\overline{C} \wedge \overline{P} \wedge R_{n+1}$.  Here
$R_{n+1}$ is a formula meant to help to exclude states that are
unreachable in $n+1$ transitions. Originally, $R_{n+1}$ is
empty. Every time a clause implied by $I_0 \wedge \Abs{T}{n+1}$ is
derived, it is added to $R_{n+1}$.  To find out if \pnt{s} is indeed
reachable, one needs to check the satisfiability of formula $I_0
\wedge \Abs{T}{n+1} \wedge \overline{\cof{A}{s}}$ where \cof{A}{s} is
the longest clause falsified by \pnt{s}. An assignment satisfying this
formula is a counterexample. If this formula is unsatisfiable, the
SAT-solver returns a clause $C^*$ implied by $I_0 \wedge \Abs{T}{n+1}$
and falsified by \pnt{s}. This clause is added to $R_{n+1}$ and \PP
looks for a new state satisfying $\overline{C} \wedge \overline{P}
\wedge R_{n+1}$. If another bad state \pnt{s} is found, \PP proceeds
as above.  Otherwise, $C$ does not specify any bad states reachable in
$(n+1)$ transitions. Then \PP picks a new clause of $H_{n+1}$ to check
if it specifies a reachable bad state.

%
%
\subsection{Description of \PP}
\label{ssec:find_bug_pcode}

\setlength{\intextsep}{2pt}
\begin{figure}
\small
\begin{tabbing}
// $\Abs{R}{n} = R_1 \wedge \dots \wedge R_n$ \\
// \\
aaa\=bb\=cc\= dd\= \kill
$\pp(I,T,P)$\{\\
\tb{\scriptsize{1}}\> $T := \mi{MakeStutter}(T)$ \\
\tb{\scriptsize{2}}\> $\mi{Cex} := \mi{Unsat}(I_0 \wedge T_{0,1} \wedge \overline{P})$ \\
\tb{\scriptsize{3}}\> if ($\mi{Cex} \neq \mi{nil}$) return($\mi{Cex},\mi{nil}$) \\
\tb{\scriptsize{4}}\> $H_1 := I_1$ \\
\tb{\scriptsize{5}}\> $n := 1$ \\
\tb{\scriptsize{6}}\> $\mi{Diam} := 1$ \\
$~~~~~~ ---------$ \\
\tb{\scriptsize{7}}\> while ($H_1 \not\equiv 1$) \{ \\
\tb{\scriptsize{8}}\Tt  if ($H_n \equiv 1$) \{ \\
\tb{\scriptsize{9}}\ttt    $n := n-1$ \\
\tb{\scriptsize{10}}\ttt     continue \}\\
\tb{\scriptsize{11}}\Tt   if $(\mi{FirstVisit}(n+1))$ $R_{n+1} := 1$ \\[8pt]
\tb{\scriptsize{12}}\Tt   $C:= \mi{PickClause}(H_n)$ \\
\tb{\scriptsize{13}}\Tt   $H_n := H_n \setminus \s{C}$ \\
\tb{\scriptsize{14}}\Tt  $H_{n+1}:= \mi{PQE}(\prob{\Abs{S}{n}}{I_0 \wedge C \wedge \Abs{H}{n} \wedge \Abs{T}{n+1}})$ \\[8pt] 
\tb{\scriptsize{15}}\Tt  $\mi{RemNoise}(H_{n+1},\Abs{R}{n+1},I,T)$\\
\tb{\scriptsize{16}}\Tt  if ($H_{n+1} \equiv 1$) continue \\
\tb{\scriptsize{17}}\Tt  $\mi{Cex} := \mi{ChkBadSt}(H_{n+1},\Abs{R}{n+1},I,T,P)$\\
\tb{\scriptsize{18}}\Tt  if ($\mi{Cex} \neq \mi{nil}$) return($\mi{Cex},\mi{nil}$) \\
\tb{\scriptsize{19}}\Tt  $n := n+1$  \\
\tb{\scriptsize{20}}\Tt  if ($\mi{Diam} < n$)  $\mi{Diam} := n$ \} \\
$~~~~~~ ---------$ \\
\tb{\scriptsize{21}}\>  return($\mi{nil},\mi{Diam}$) \}\\
\end{tabbing} 
\vspace{-15pt}
\caption{The \PP procedure}
\label{fig:prove_prop}
\end{figure}

Pseudo-code of \PP is given in Figure~\ref{fig:prove_prop}. \PP
accepts formulas $I$, $T$ and $P$ specifying initial states,
transition relation and the property to be verified respectively. \PP
returns a counterexample if $P$ does not hold, or \di{I,T} if $P$
holds. As we mentioned above, to simply compute the value of \di{I,T},
it suffices to call \PP with the trivial property $P$ that is always
true.

\PP consists of three parts separated by the dotted line. The
first part (lines 1-6) starts with modifying the transition relation
to introduce stuttering (see Subsection~\ref{ssec:stuttering}). Then
\PP checks if a bad state is reachable in one transition (lines 2-3).
After that, \PP sets formula $H_1$ to $I_1$ and parameter $n$ to 1
(lines 4-5).  The parameter $n$ stores the index of the latest time
frame where the corresponding formula $H_n$ is not empty. \PP
concludes the first part by setting the value of the diameter to 1
(line 6).

The second part consists of a \ti{while} loop (lines 7-20). In this
loop, \PP pushes formula $I_1$ and its descendants to later time
frames.  This part consists of three pieces separated by vertical
spaces.  The first piece (lines 8-11) starts by checking if formula
$H_n$ has no clauses (and so $H_n \equiv 1$).  If this is the case,
then all descendants of $H_n$ have been proved redundant.  So \PP
decreases the value of $n$ by 1 and starts a new iteration.  If $H_n
\not\equiv 1$, \PP checks if \mbox{$(n+1)$-th} time frame is visited
for the first time. If so, \PP sets formula $R_{n+1}$ to 1.  As we
mentioned in the previous subsection, $R_{n+1}$ is used to accumulate
clauses implied by $I_0 \wedge \Abs{T}{n+1}$.

\PP starts the second piece of the \ti{while} loop (lines 12-14) by
picking a clause $C$ of formula $H_n$ and removing it from
$H_n$. After that, \PP builds formula $H_{n+1}$ such that
$\prob{\Abs{S}{n}}{I_0 \wedge \Abs{H}{n} \wedge C \wedge \Abs{T}{n+1}}
\equiv H_{n+1} \wedge \prob{\Abs{S}{n}}{I_0 \wedge \Abs{H}{n} \wedge
  \Abs{T}{n+1}}$.

In the third piece, (lines 15-20), \PP analyzes formula
$H_{n+1}$. First, it calls procedure \ti{RemNoise} described in the
next subsection. It drops noise clauses of $H_{n+1}$ i.e. ones implied
by $I_0 \wedge \Abs{T}{n+1}$.  If the resulting formula $H_{n+1}$ is
empty, \PP starts a new iteration.  Otherwise, \PP calls procedure
\ti{ChkBadSt} also described in the next subsection.  \ti{ChkBadSt}
checks if clauses of $H_{n+1}$ exclude a bad state reachable in $n+1$
transitions. If not, i.e. if no counterexample is found, \PP
increments the value of $n$ by 1. If the value of $n$ is greater than
the current diameter $\mi{Diam}$, the latter is set to $n$ (line
20). After that a new iteration begins.

The third part of \PP consists of line 21. \PP gets to this line if
$I_1$ is proved redundant and no bad state is reachable in \di{I,T}
transitions. This means that property $P$ holds and \PP returns the
value of \di{I,T}.

%
%
\subsection{Description of RemNoise and ChkBadSt procedures}
\label{ssec:two_routines}
\setlength{\intextsep}{2pt}
\begin{figure}
\small
\begin{tabbing}
aaa\=bb\=cc\= dd\= \kill
$\mi{RemNoise}(H_i,\Abs{R}{i},I,T)$\{\\
\tb{\scriptsize{1}}\> for each clause $C \in H_i$ \{ \\
\tb{\scriptsize{2}}\Tt if $(\mi{Unsat}(I_0 \wedge \Abs{T}{i} \wedge \Abs{R}{i} \wedge \overline{C}))$ \{ \\
\tb{\scriptsize{3}}\ttt  $H_i := H_i \setminus \s{C}$\\
\tb{\scriptsize{4}}\ttt  $R_i := R_i \wedge C$ \}\}     \\
\end{tabbing} 
\vspace{-15pt}
\caption{The $\mi{RemNoise}$ procedure}
\vspace{7pt}
\label{fig:rem_noise}
\end{figure}

%
Pseudo-code of \ti{RemNoise} is given in Figure~\ref{fig:rem_noise}.
The objective of \ti{RemNoise} is to remove noise clauses of $H_i$
i.e.  ones implied by $I_0 \wedge \Abs{T}{i}$. So for every clause $C$
of $H_i$, \ti{RemNoise} checks if formula $I_0 \wedge \Abs{T}{i}
\wedge \Abs{R}{i} \wedge \overline{C}$ is satisfiable. (Here
\Abs{R}{i}= $R_1 \wedge \dots \wedge R_i$. It specifies clauses
implied by $I_0 \wedge \Abs{T}{i}$ that have been generated earlier.)
If the formula above is unsatisfiable, $C$ is removed from $H_i$ and
added to $R_i$.

\setlength{\intextsep}{2pt}
\begin{figure}
\small
\begin{tabbing}
aaa\=bb\=cc\= dd\= \kill
$\mi{ChkBadSt}(H_i,\Abs{R}{i},I,T,P)$\{\\
\tb{\scriptsize{1}}\> for each clause $C \in H_i$ \{\\
\tb{\scriptsize{2}}\Tt while ($\mi{true}$) \{ \\
\tb{\scriptsize{3}}\ttt  $\pnt{s} := \mi{SatAssgn}(\overline{C} \wedge \overline{P} \wedge R_i)$ \\
\tb{\scriptsize{4}}\ttt  if ($\pnt{s} = \mi{nil}$) break\\
\tb{\scriptsize{5}}\ttt  $(\mi{Cex},C^*) := \mi{Unsat}(I_0 \wedge \Abs{T}{i} \wedge \Abs{R}{i} \wedge \overline{\cof{A}{s}})$  \\
\tb{\scriptsize{6}}\ttt  if ($\mi{Cex} \neq \mi{nil}$) return($\mi{Cex}$) \\
\tb{\scriptsize{7}}\ttt  $R_i := R_i \wedge C^*$\}\} \\
\tb{\scriptsize{8}}\> return($\mi{nil}$)\}\\
\end{tabbing} 
\vspace{-15pt}
\caption{The $\mi{ChkBadSt}$ procedure}
\vspace{7pt}
\label{fig:ChkBadSt}
\end{figure}

%
Pseudo-code of \ti{ChkBadSt} is given in Figure~\ref{fig:ChkBadSt}.
It checks if a clause of $H_i$ specifies a bad state reachable in $i$
transitions for the first time. The idea of \ti{ChkBadSt} was
described in Subsection~\ref{ssec:find_cex}. \ti{ChkBadSt} consists of
two nested loops. In the outer loop, \ti{ChkBadSt} enumerates clauses
of $H_i$. In the inner loop, \ti{ChkBadSt} checks if a bad state
\pnt{s} satisfying formula $\overline{C} \wedge \overline{P} \wedge
R_i$ is reachable in $i$ transitions. The inner loop iterates until
this formula becomes unsatisfiable.

Finding out if \pnt{s} is reachable in $i$ transitions comes down to
checking the satisfiability of formula $I_0 \wedge \Abs{T}{i} \wedge
\Abs{R}{i} \wedge \overline{\cof{A}{s}}$.  (Here \cof{A}{s} is the
longest clause falsified by \pnt{s}.) An assignment satisfying this
formula specifies a counterexample. If this formula is unsatisfiable,
a clause $C^*(S_i)$ is returned that is implied by $I_0 \wedge
\Abs{T}{i}$ and falsified by \pnt{s}. This clause is added to $R_i$
and a new iteration of the inner loop begins.

\section{The \FC procedure}
\label{sec:prove_prop*}
\setlength{\intextsep}{2pt}
\begin{figure}
\small
\begin{tabbing}
aaa\=bb\=cc\= dd\= \kill
$\fc(I,T,P)$\{\\
\tb{\scriptsize{1}}\> $T := \mi{MakeStutter}(T)$\\
\tb{\scriptsize{2}}\> $\mi{Cex} := \mi{Unsat}(I_0 \wedge T_{0,1} \wedge \overline{P})$ \\
\tb{\scriptsize{3}}\> if ($\mi{Cex} \neq \mi{nil}$) return($\mi{Cex}$) \\
\tb{\scriptsize{4}}\> $\Sup{I}{exp} := \mi{ExpandInitStates}(I,P)$ \\
$~~~~~~~  ---------$ \\
\tb{\scriptsize{5}}\> while ($\mi{true}$) \{\\
\tb{\scriptsize{6}}\Tt $(\mi{Cex},\mi{Diam}) := \pp(\Sup{I}{exp},T,P)$  \\
\tb{\scriptsize{7}}\Tt if ($\mi{Cex} \neq \mi{nil}$) \{ \\
\tb{\scriptsize{8}}\ttt   $\pnt{s_0} := \mi{ExtractInitState}(\mi{Cex})$ \\
\tb{\scriptsize{9}}\ttt    if ($I(\pnt{s_0}) = 1$) return($\mi{Cex}$) \\
\tb{\scriptsize{10}}\ttt    $\mi{ExcludeState}(\Sup{I}{exp},\pnt{s_0})$ \\
\tb{\scriptsize{11}}\ttt    continue~~\} \\
\tb{\scriptsize{12}}\Tt return($\mi{nil}$)\}\} \\

\end{tabbing} 
\vspace{-15pt}
\caption{The \FC procedure}
\vspace{7pt}
\label{fig:prove_prop*}
\end{figure}

%
When a property holds, the \PP procedure described in
Section~\ref{sec:prove_prop} has to examine traces of length up to the
reachability diameter. This strategy may be inefficient for transition
systems with a large diameter. In this section, we describe a
variation of \PP called \FC that addresses this problem.  In
particular, \FC can prove a property by examining traces that are much
shorter than the diameter. The main idea of \FC is to expand the set
of initial states by adding $P$-states that may not be reachable at
all.  So faster convergence is achieved by expanding the set of
allowed behaviors.  This is similar to boosting the performance of
existing methods of property checking by looking for a weaker
inductive invariant (as opposed to building the strongest inductive
invariant satisfied only by \ti{reachable} states).

The pseudo-code of \FC is given in Figure~\ref{fig:prove_prop*}.  It
consists of two parts separated by the dotted line. \FC starts the
first part (lines 1-4) by introducing stuttering. Then it checks if
there is a bad state reachable in one transition. Finally, it
generates a formula \Sup{I}{exp} specifying an expanded set of initial
states that satisfies \impl{I}{\Sup{I}{exp}} and
\impl{\Sup{I}{exp}}{P} (line 4).  Here $I$ is the initial set of
states and $P$ is the property to be proved. A straightforward way to
generate \Sup{I}{exp} is to simply set it to $P$.

The second part (lines 5-12) consists of a \ti{while} loop. In this
loop, \FC repeatedly calls the \PP procedure described in
Section~\ref{sec:prove_prop} (line 6). It returns
$(\mi{Cex},\mi{Diam})$.  If $\mi{Cex} = \mi{nil}$, property $P$ holds
and \FC returns $\mi{nil}$ (line 12).  Otherwise, \FC analyzes the
counterexample $\mi{Cex}=\tr{s}{0}{n}$ returned by \PP (lines 7-11).
If state \pnt{s_0} of $\mi{Cex}$, satisfies $I$, then $P$ does not
hold and \FC returns $\mi{Cex}$ as a counterexample (line 9).  If
$I(\pnt{s_0})=0$, \FC excludes \pnt{s_0} by conjoining \Sup{I}{exp}
with a clause $C$ such that $C(\pnt{s_0})=0$ and \impl{I}{C}. Then \FC
starts a new iteration. When constructing clause $C$ it makes sense to
analyze $\mi{Cex}$ to find other states of \Sup{I}{exp} to be
excluded. Suppose, for instance, that one can easily prove that state
\pnt{s_1} of $\mi{Cex}$ can be reached from a state \pnt{s^*_0} such
that \Sup{I}{exp}(\pnt{s^*_0})=1, $I(\pnt{s^*_0})=0$ and $\pnt{s^*_0}
\neq \pnt{s_0}$. Then one may try to pick clause $C$ so that it is
falsified by both \pnt{s_0} and \pnt{s^*_0}.

\FC is a complete procedure i.e. it eventually proves $P$ or finds a
counterexample.

\section{Some Background}
\label{sec:background}

The first methods of property checking were based on BDDs and computed
the set of reachable states~\cite{mc_thesis}.  Since BDDs frequently
get prohibitively large, SAT-based methods of property checking have
been introduced. Some of them, like interpolation~\cite{ken03} and
IC3~\cite{ic3} have achieved a great boost in performance.  Among
incomplete SAT-based methods (that can do only bug hunting), Bounded
Model Checking (BMC)~\cite{bmc} has enjoyed a lot of success.

As we mentioned in the introduction, the problem with inductive
invariants is that they can be too large to generate or too hard to
find.  Besides, if a property is false due to a deep bug, looking for
an inductive invariant may not be the best strategy to find this bug.
After the introduction of PQE~\cite{hvc-14}, we formulated a few
approaches addressing the problems above. In particular,
in~\cite{tech_rep_crr}, we described a PQE-based procedure for
property checking meant for finding deep bugs. However, that procedure
was incomplete.  Here, we continue this line of research. Similarly to
the procedure of~\cite{tech_rep_crr}, \PP performs depth-first search
meant to facilitate finding deep bugs. However, in contrast to the
former, \PP is complete.

The idea of proving a property without generating an inductive
invariant is not new. For instance, earlier it was proposed to combine
BMC with finding a \ti{recurrence} diameter~\cite{recurr_diam}. The
latter is equal to the length of the longest trace that does not
repeat a state. Obviously, the recurrence diameter is larger or equal
to the reachability diameter. In particular, the former can be
drastically larger than the latter. In this case, finding the
recurrence diameter is of no use.

\section{Conclusions}
\label{sec:conclusions}
In this paper, we present \pp, a new procedure for checking safety
properties.  It is based on a technique called Partial Quantifier
Elimination (PQE). In contrast to regular quantifier elimination, in
PQE, only a small part of the formula is taken out of the scope of
quantifiers. In~\cite{fmcad12,fmcad13,cert_tech_rep}, we developed the
machinery of redundancy based reasoning meant for building efficient
PQE solvers. The advantage of \PP is twofold.  First, it can prove
that a property holds without generation of an inductive
invariant. This can be very useful when inductive invariants are
prohibitively large or are hard to find. Second, \PP performs
depth-first search and so can be used for finding deep bugs.

\bibliographystyle{IEEEtran}
\bibliography{short_sat,local}
\vspace{15pt}
\appendix
\setcounter{proposition}{0}
\section{Proofs}
Lemma~\ref{lemma:diam} below is used in proving
Proposition~\ref{prop:diam_pqe}.
%
%
\begin{lemma}
\label{lemma:diam}
Let \KS be an $(I,T)$-system. Then $\di{I,T} \leq n$ iff
$\prob{\Abs{S}{n}}{I_0 \wedge \Abs{T}{n+1}} \equiv
\prob{\Abs{S}{n}}{I_1 \wedge \Abs{T}{n+1}}$ where $n \geq 0$.
\end{lemma}
\begin{proof}[\kern-10pt Proof]
\tb{If part:} Given $\prob{\Abs{S}{n}}{I_0 \wedge \Abs{T}{n+1}} \equiv
\prob{\Abs{S}{n}}{I_1 \wedge \Abs{T}{n+1}}$, let us prove $\di{I,T}\!
\leq\!n$. Assume the contrary, i.e. $\di{I,T}\!>\!n$. Then there is a
state \pnt{a_{n+1}} reachable only in $(n\!+\!1)$-th time frame.
Hence, there is a trace $t_a\!=\!\Tr{a}{0}{n+1}$ satisfying
$I_0\!\wedge\!  \Abs{T}{n+1}$. Then due to
$\prob{\Abs{S}{n}}{I_0\!\wedge
  \Abs{T}{n+1}}\!\equiv\!\prob{\Abs{S}{n}}{I_1\!\wedge\!\Abs{T}{n+1}}$
there exists a trace $t_b\!=\!\Tr{b}{0}{n+1}$ satisfying $I_1\!
\wedge\!\Abs{T}{n+1}$ where \pnt{b_{n+1}}= \pnt{a_{n+1}}.

Let $t_c$ be a trace \tr{c}{0}{n} where $\pnt{c_i} = \pnt{b_{i+1}}$,
$i=0,\dots,n$. The fact that $t_b$ satisfies $I_1 \wedge \Abs{T}{n+1}$
implies that $t_c$ satisfies $I_0 \wedge \Abs{T}{n}$. Since \pnt{c_n}
= \pnt{b_{n+1}} = \pnt{a_{n+1}}, state \pnt{a_{n+1}} is reachable in
$n$ transitions. So we have a contradiction.

\tb{Only if part:} Given $\di{I,T} \leq n$, let us prove that 
$\prob{\Abs{S}{n}}{I_0 \wedge \Abs{T}{n+1}} \equiv
\prob{\Abs{S}{n}}{I_1 \wedge \Abs{T}{n+1}}$.

First, let us show that \prob{\Abs{S}{n}}{I_0 \wedge \Abs{T}{n+1}}
implies \prob{\Abs{S}{n}}{I_1 \wedge \Abs{T}{n+1}}.  Let
\prob{\Abs{S}{n}}{I_0 \wedge \Abs{T}{n+1}}=1 under an assignment
\pnt{s_{n+1}} to $S_{n+1}$. Then the state \pnt{s_{n+1}} is reachable
in $n+1$ transitions. Since $\di{I,T} \leq n$, there has to be a trace
$t_a = \tr{a}{0}{k}$ where $k \leq n$ and \pnt{a_k} =
\pnt{s_{n+1}}. Let $m$ be equal to $n+1-k$.  Let $t_b =
\tr{b}{0}{n+1}$ be a trace defined as follows: \pnt{b_i}=\pnt{a_0},
$i=0,\dots,m$, and \pnt{b_i}=\pnt{a_{i-m}}, $i=m+1,\dots,n+1$. Due to
the stuttering feature of \ks, the trace $t_b$ satisfies $I_1 \wedge
\Abs{T}{n+1}$ and \pnt{b_{n+1}}=\pnt{a_k}=\pnt{s_{n+1}}. So,
\prob{\Abs{S}{n}}{I_1 \wedge \Abs{T}{n+1}}=1 under the assignment
\pnt{s_{n+1}} to $S_{n+1}$.

Now, we show that \prob{\Abs{S}{n}}{I_1 \wedge \Abs{T}{n+1}} implies
$\pprob{\Abs{S}{n}}{I_0 \wedge \Abs{T}{n+1}}$. Let
\prob{\Abs{S}{n}}{I_1 \wedge \Abs{T}{n+1}}=1 under an assignment
\pnt{s_{n+1}} to $S_{n+1}$. Then \pnt{s_{n+1}} is reachable in $n$
transitions. Due to the stuttering feature of \ks, the state
\pnt{s_{n+1}} is also reachable in $n\!+\!1$ transitions. So, a trace
\tr{s}{0}{n+1} satisfies $I_0 \wedge \Abs{T}{n+1}$.  Hence,
\prob{\Abs{S}{n}}{I_0 \wedge \Abs{T}{n+1}}=1 under the assignment
\pnt{s_{n+1}}.
\end{proof}
%
%
\begin{proposition}
Let \KS be an $(I,T)$-system. Then $\di{I,T} \leq n$ iff formula $I_1$
is redundant in $\pprob{\Abs{S}{n}}{I_0 \wedge I_1 \wedge
  \Abs{T}{n+1}}$ (i.e.  iff $\pprob{\Abs{S}{n}}{I_0 \wedge
  \Abs{T}{n+1}} \equiv \pprob{\Abs{S}{n}}{I_0 \wedge I_1 \wedge
  \Abs{T}{n+1}}$).
\end{proposition}
\begin{proof}[\kern-10pt Proof]
Lemma~\ref{lemma:diam} entails that to prove the proposition at hand
it is sufficient to show that $\pprob{\Abs{S}{n}}{I_0 \wedge
  \Abs{T}{n+1}} \equiv \pprob{\Abs{S}{n}}{I_1 \wedge \Abs{T}{n+1}}$
iff formula $I_1$ is redundant in $\pprob{\Abs{S}{n}}{I_0 \wedge I_1
  \wedge \Abs{T}{n+1}}$.

\tb{If part:} Given $I_1$ is redundant in \prob{\Abs{S}{n}}{I_0 \wedge
  I_1 \wedge \Abs{T}{n+1}}, let us show that $\prob{\Abs{S}{n}}{I_0
  \wedge \Abs{T}{n+1}} \equiv \prob{\Abs{S}{n}}{I_1 \wedge
  \Abs{T}{n+1}}$.  Redundancy of $I_1$ means that
$\prob{\Abs{S}{n}}{I_0 \wedge I_1 \wedge \Abs{T}{n+1}} \equiv
\prob{\Abs{S}{n}}{I_0 \wedge \Abs{T}{n+1}}$. Let us show that $I_0$ is
redundant in $\prob{\Abs{S}{n}}{I_0 \wedge I_1 \wedge \Abs{T}{n+1}}$
and hence $\prob{\Abs{S}{n}}{I_1 \wedge \Abs{T}{n+1}} \equiv
\prob{\Abs{S}{n}}{I_0 \wedge \Abs{T}{n+1}}$.  Assume the contrary
i.e. $I_0$ is not redundant and hence $\prob{\Abs{S}{n}}{I_0 \wedge
  I_1 \wedge \Abs{T}{n+1}}$ $\not\equiv$ $\prob{\Abs{S}{n}}{I_1 \wedge
  \Abs{T}{n+1}}$. Then there is an assignment \pnt{s_{n+1}} to
variables of $S_{n+1}$ for which $\prob{\Abs{S}{n}}{I_1 \wedge
  \Abs{T}{n+1}} = 1$ and $\prob{\Abs{S}{n}}{I_0 \wedge I_1 \wedge
  \Abs{T}{n+1}} = 0$. (The opposite is not possible since $I_0 \wedge
I_1 \wedge \Abs{T}{n+1}$ implies $I_1 \wedge \Abs{T}{n+1}$.) This
means that
\begin{itemize}
\item there is a valid trace $t_a$= \tr{a}{0}{n+1} where \pnt{a_1} satisfies $I_1$ and \pnt{a_{n+1}} = \pnt{s_{n+1}}.
\item there is no trace $t_b$=\tr{b}{0}{n+1} where \pnt{b_0} satisfies $I_0$, \pnt{b_1} satisfies $I_1$
and \pnt{b_{n+1}} = \pnt{s_{n+1}}.
\end{itemize}
Let us pick $t_b$ as follows. Let \pnt{b_k}=\pnt{a_k} for \mbox{$1
  \leq k \leq n+1$} and \pnt{b_0}=\pnt{b_1}.  Let us show that $t_b$
satisfies $I_0 \wedge I_1 \wedge \Abs{T}{n+1}$ and so we have a
contradiction. Indeed, \pnt{b_0} satisfies $I_0$ because \pnt{b_1}
satisfies $I_1$ and \pnt{b_0}=\pnt{b_1}. Besides,
(\pnt{b_0},\pnt{b_1}) satisfies $T_{0,1}$ because the system at hand
has the stuttering feature. Hence $t_b$ satisfies $I_0 \wedge I_1
\wedge \Abs{T}{n+1}$.

\tb{Only if part:} Given $\prob{\Abs{S}{n}}{I_0 \wedge \Abs{T}{n+1}}
\equiv \prob{\Abs{S}{n}}{I_1 \wedge \Abs{T}{n+1}}$, let us show that
$I_1$ is redundant in \prob{\Abs{S}{n}}{I_0 \wedge I_1 \wedge
  \Abs{T}{n+1}}. Assume the contrary i.e.  $\prob{\Abs{S}{n}}{I_0
  \wedge \Abs{T}{n+1}} \not\equiv \prob{\Abs{S}{n}}{I_0 \wedge I_1
  \wedge \Abs{T}{n+1}}$.  Then there is an assignment \pnt{s_{n+1}} to
variables of $S_{n+1}$ such that $\prob{\Abs{S}{n}}{I_0 \wedge
  \Abs{T}{n+1}} = 1$ and $\prob{\Abs{S}{n}}{I_0 \wedge I_1 \wedge
  \Abs{T}{n+1}} = 0$. This means that
\begin{itemize}
\item there is a valid trace $t_a$= \tr{a}{0}{n+1} where \pnt{a_0} satisfies $I_0$ and \pnt{a_{n+1}} = \pnt{s_{n+1}}
\item there is no trace $t_b$=\tr{b}{0}{n+1} where \pnt{b_0} satisfies $I_0$, \pnt{b_1} satisfies $I_1$
and \pnt{b_{n+1}} = \pnt{s_{n+1}}.
\end{itemize}
Let us show that then \prob{\Abs{S}{n}}{I_1 \wedge \Abs{T}{n+1}}
evaluates to 0 for \pnt{s_{n+1}}. Indeed, assume the contrary i.e.
there is an assignment $t_c = \tr{c}{0}{n+1}$ satisfying $I_1 \wedge
\Abs{T}{n+1}$ where \pnt{c_1} satisfies $I_1$ and \pnt{c_{n+1}} =
\pnt{s_{n+1}}. Let $t_d=\tr{d}{0}{n+1}$ be obtained from $t_c$ as
follows: \pnt{d_0}=\pnt{d_1}=\pnt{c_1}, \pnt{d_i}=\pnt{c_i},
$i=2,\dots,n+1$.  Then $t_d$ satisfies $I_0 \wedge I_1 \wedge
\Abs{T}{n+1}$ which contradicts the claim above that there is no trace
$t_b$. Hence, \prob{\Abs{S}{n}}{I_0 \wedge \Abs{T}{n+1}}=1 and
\prob{\Abs{S}{n}}{I_1 \wedge \Abs{T}{n}}=0 under assignment
\pnt{s_{n+1}}. So we have a contradiction.
\end{proof}

\vspace{4pt}
%
%
\begin{proposition}
Let \KS be an $(I,T)$-system and $P$ be a property of \ks. No
$\overline{P}$-state is reachable in $(n+1)$-th time frame for the
first time iff $I_1$ is redundant in $\prob{\Abs{S}{n}}{I_0 \wedge I_1
  \wedge \Abs{T}{n+1} \wedge \overline{P}}$.
\end{proposition}
\begin{proof}[\kern-10pt Proof]
\tb{If part:} Assume the contrary i.e. $I_1$ is redundant in
$\prob{\Abs{S}{n}}{I_0 \wedge I_1 \wedge \Abs{T}{n+1} \wedge
  \overline{P}}$ but there is a bad state \pnt{s_{n+1}} that is
reachable in $(n+1)$-th time frame for the first time. Then there is
an assignment $t_a = \tr{a}{0}{n+1}$ satisfying $I_0 \wedge
\Abs{T}{n+1} \wedge \overline{P}$ where \pnt{a_{n+1}} = \pnt{s_{n+1}}.
Redundancy of $I_1$ means that $\prob{\Abs{S}{n}}{I_0 \wedge I_1
  \wedge \Abs{T}{n+1} \wedge \overline{P}} \equiv
\prob{\Abs{S}{n}}{I_0 \wedge \Abs{T}{n+1} \wedge \overline{P}}$.  Then
there is an assignment $t_b = \tr{b}{0}{n+1}$ where \pnt{b_{n+1}} =
\pnt{s_{n+1}} that satisfies $I_0 \wedge I_1 \wedge \Abs{T}{n+1}
\wedge \overline{P}$. Let $t_c = \tr{c}{0}{n}$ where $\pnt{c_i} =
\pnt{b_{i+1}}$, $i=0,\dots,n$. Then $I(\pnt{c_0})=1$ and
$P(\pnt{c_n})=0$ since \pnt{c_n}=\pnt{b_{n+1}}=\pnt{s_{n+1}}. The fact
that $t_c$ is a valid trace entails that the state \pnt{s_{n+1}} is
reachable in $n$-th time frame as well. So we have a contradiction.

\tb{Only if part:} Assume the contrary i.e. no bad state is reachable
in $(n+1)$-th time frame for the first time but $I_1$ is not redundant
in $\prob{\Abs{S}{n}}{I_0 \wedge I_1 \wedge \Abs{T}{n+1} \wedge
  \overline{P}}$.  This means that $\prob{\Abs{S}{n}}{I_0 \wedge I_1
  \wedge \Abs{T}{n+1} \wedge \overline{P}} \not\equiv
\prob{\Abs{S}{n}}{I_0 \wedge \Abs{T}{n+1} \wedge \overline{P}}$.  Then
there is an assignment \pnt{s_{n+1}} to variables of $S_{n+1}$ such
that \prob{\Abs{S}{n}}{I_0 \wedge \Abs{T}{n+1} \wedge \overline{P}}=1
and \prob{\Abs{S}{n}}{I_0 \wedge I_1 \wedge \Abs{T}{n+1} \wedge
  \overline{P}}=0 under \pnt{s_{n+1}}.  The means that there is an
assignment \tr{a}{0}{n+1} satisfying $I_0 \wedge \Abs{T}{n+1} \wedge
\overline{P}$ where \pnt{a_{n+1}} = \pnt{s_{n+1}}.  Hence,
\pnt{s_{n+1}} is a bad state that is reachable in $(n+1)$-th time
frame.

Let us show that \pnt{s_{n+1}} is not reachable in a previous time
frame. Assume the contrary i.e. \pnt{s_{n+1}} is reachable in $k$-th
time frame where $k < n+1$.  Then there is an assignment
\linebreak$t_b = \tr{b}{0}{k}$ satisfying $I_0 \wedge \Abs{T}{k}
\wedge \overline{P}$ where \pnt{b_k}=\pnt{s_{n+1}}. Let $t_c =
\tr{c}{0}{n+1}$ be defined as follows: \pnt{c_0}=\pnt{c_1}=\pnt{b_0},
\mbox{\pnt{c_i}=\pnt{b_{i-1}}}, $i=2,\dots,k+1$, \pnt{c_i} =
\pnt{b_k}, $i = k+2,\dots,n+1$. Informally, $t_c$ specifies the same
sequence of states as $t_b$ plus stuttering in the initial state and
after reaching state \pnt{c_{k+1}} equal to \pnt{b_k} (and so to
\pnt{s_{n+1}}). Then $t_c$ satisfies $I_0 \wedge I_1 \wedge
\Abs{T}{n+1} \wedge \overline{P}$ under assignment \pnt{s_{n+1}} and
we have a contradiction.
\end{proof}
%
%
\begin{proposition}
\label{prop:dfs}
Let \KS be an $(I,T)$-system and $H(S_{n+1})$ be a formula. Then $I_0
\wedge I_1 \wedge \Abs{T}{n+1} \rightarrow H$ \linebreak entails $I_1 \wedge
T_{1,2} \wedge \dots \wedge T_{n,n+1} \rightarrow H$.
\end{proposition}
\begin{proof}
Assume that $I_1 \wedge T_{1,2} \wedge \dots \wedge T_{n,n+1}
\rightarrow H$ does not hold. Then there is a trace
$t_a$=\tr{a}{1}{n+1} that satisfies\linebreak$I_1 \wedge T_{1,2}
\wedge \dots \wedge T_{n,n+1}$ but falsifies $H$. The latter means
that \pnt{a_{n+1}} falsifies $H$. Let trace $t_b = \tr{b}{0}{n+1}$ be
obtained from $t_a$ as follows: $\pnt{b_0} = \pnt{b_1}$,
\pnt{b_i}=\pnt{a_i}, $i=1,\dots,n+1$.  Since $I(\pnt{a_1})=1$, then
$I(\pnt{b_0})=I(\pnt{b_1})=1$.  Due to the stuttering feature,
$T(\pnt{b_0},\pnt{b_1})=1$. So trace $t_b$ satisfies $I_0 \wedge I_1
\wedge \Abs{T}{n+1}$. Since \mbox{$\pnt{b_{n+1}} = \pnt{a_{n+1}}$},
then $H(\pnt{b_{n+1}}) = 0$ and \mbox{$I_0 \wedge I_1 \wedge
  \Abs{T}{n} \not\rightarrow H$}. So we have a contradiction.
\end{proof}

\end{document}